\documentclass[a4paper,USenglish,cleveref]{lipics-v2021}

\hideLIPIcs
\nolinenumbers

\usepackage{caption}
\usepackage{subcaption}

\newcommand{\psfrage}[1]{{\color{blue}{\sf[PS: #1]}}}
\newcommand{\hpfrage}[1]{{\color{violet}\sf[HP: #1]}}
\newcommand{\swfrage}[1]{{\color{teal}\sf[SW: #1]}}
\renewcommand{\psfrage}[1]{} \renewcommand{\hpfrage}[1]{} \renewcommand{\swfrage}[1]{}

\usepackage{xcolor}
\usepackage{xspace}
\usepackage{booktabs}
\usepackage{dsfont}
\usepackage{footmisc}
\usepackage{marvosym}
\usepackage{amsmath}
\usepackage{hyperref}
\usepackage[capitalise,noabbrev]{cleveref}
\usepackage{tabularx}
\usepackage{listings}
\usepackage{multirow}

\usepackage{makecell}
\usepackage{placeins}
\usepackage{dblfloatfix}
\usepackage{midfloat}
\usepackage{afterpage}
\newcommand{\seq}[1]{\left\langle #1\right\rangle}
\newcommand{\seqGilt}[2]{\left\langle #1\gilt #2\right\rangle}
\newcommand{\Id}[1]{\ensuremath{\text{{\sf #1}}}}

\newcommand{\set}[1]{\left\{ #1\right\}}
\newcommand{\gilt}{:}

\newcommand{\setGilt}[2]{\left\{ #1\gilt #2\right\}}

\newcommand{\binomial}[2]{\binom{#1}{#2}}

\newcommand{\nat}{\mathds{N}}

\newcommand{\nplus}{\mathbb{N}_+}

\newcommand{\realrange}[2]{\left[#1, #2\right]}

\newcommand{\unitrange}[2]{\realrange{0}{1}}

\newcommand{\Oh}[1]{\mathcal{O}\!\left( #1\right)}

\newcommand{\Ohsmall}[1]{\mathcal{O}(#1)}

\newcommand{\Th}[1]{\Theta\!\left( #1\right)}
\newcommand{\Om}[1]{\Omega\left(#1\right)}
\newcommand{\Omsmall}[1]{\Omega(#1)}

\newcommand{\discussionsize}{\small}

\newcommand{\labelcommand}{}
\newsavebox{\buchalgorithmparam}
\newcounter{lineNumber}
\newenvironment{buchalgorithmpos}[3]{%
\renewcommand{\labelcommand}{#2}%
\renewcommand{\captiontext}{#3}%
\sbox{\buchalgorithmparam}{\parbox{\textwidth}{#3}}%
\begin{figure}[#1]\begin{code}\setcounter{lineNumber}{1}}
{\end{code}%
\caption{\label{\labelcommand}\captiontext}\end{figure}}

{\end{buchalgorithmpos}}

\newenvironment{code}{\noindent%
\begin{tabbing}%
\hspace{1.5em}\=\hspace{1.5em}\=\hspace{1.5em}\=\hspace{1.5em}\=\hspace{1.5em}\=\hspace{1.5em}\=\hspace{1.5em}\=%
\hspace{1.5em}\=\hspace{1.5em}\=\hspace{1.5em}\=\hspace{1.5em}\=\hspace{1.5em}\=%
\kill}{\end{tabbing}}

\newcommand{\Class}{{\bf Class\ }}

\newcommand{\Array}{{\Id{Array}\ }}
\newcommand{\Of}{\ensuremath{\text{\bf of\ }}}
\newcommand{\Function} {{\bf Function\ }}
\newcommand{\Funct}[3]{\Function #1\Declare{{\rm (}{#2\rm )}}{#3}}

\newcommand{\Procedure}{{\bf Procedure\ }}

\newlength{\mynegthinspace}
\newlength{\mysmallspace}
\newcommand{\While}    {{\bf while\ }}

\newcommand{\Do}       {{\bf do\ }}

\newcommand{\For}      {{\bf for\ }}

\newcommand{\Is}{\ensuremath{\mathbin{:=}}}

\newcommand{\ForFromTo}[3]{{\For $#1 \Is #2$ \To $#3$ \Do}}

\newcommand{\To}       {{\bf to\ }}

\newcommand{\If}       {{\bf if\ }}

\newcommand{\Then}     {{\bf then\ }}
\newcommand{\Else}     {{\bf else\ }}

\newcommand{\Return}   {{\bf return\ }}

\newcommand{\Boolean}  {\ensuremath{\set{\True,\False}}}

\newcommand{\True}     {{\bf true}}
\newcommand{\False}    {{\bf false}}

\newcommand{\Not}       {\ensuremath{\neg}}
\newcommand{\Or}       {\ensuremath{\vee}}
\newcommand{\Div}       {{\bf\ div\ }}
\newcommand{\Increment}{\raisebox{.12ex}{\hbox{\tt ++}}}
\newcommand{\Decrement}{\raisebox{.12ex}{\hbox{\tt -}{\tt -}}}

\newcommand{\Rem}[1]   {{\bf //\hspace{0.5mm}{\rm#1}}}
\newcommand{\RRem}[1]   {\`{\bf //\hspace{0.5mm}~}{\rm#1}}

\newcommand{\Declare}[2]{#1\mbox{ \rm : }#2}
\newcommand{\DeclareInit}[3]{#1\ensuremath{{}={}}#3\mbox{ \rm : }#2} %Mathemodus außenherum möglich

\newdimen\endofsize\endofsize=0.5em

\crefname{listing}{Algorithm}{Algorithms}
\crefname{lstlisting}{Algorithm}{Algorithms}
\Crefname{lstlisting}{Algorithm}{Algorithms}
\crefname{@theorem}{Theorem}{Theorems}
\Crefname{@theorem}{Theorem}{Theorems}

\newcommand{\myparagraph}[1]{\subparagraph*{#1}}
\newcommand{\anonymous}[1]{\ifx\authoranonymous\relax\textcolor{red}{Anonymous}\else{#1}\fi}
\let\oldcite\cite
\renewcommand\cite{\unskip~\oldcite}
\newcommand{\backyard}{\ensuremath{T'}}
\newcommand{\maxOffset}{\hat{o}}
\newcommand{\maxB}{\hat{B}}
\newcommand{\threshold}{\ensuremath{\delta}}
\newcommand{\maxThreshold}{\ensuremath{\hat{t}}}
\newcommand{\meta}{\ensuremath{M}}

\newcommand{\mytitlerunning}{Sliding Block Hashing (Slick) -- Basic Algorithmic Ideas}
\title{Sliding Block Hashing (Slick)\\ -- Basic Algorithmic Ideas}
\titlerunning{\mytitlerunning}

\author{Hans-Peter Lehmann}{Karlsruhe Institute of Technology, Germany}{hans-peter.lehmann@kit.edu}{https://orcid.org/0000-0002-0474-1805}{}
\author{Peter Sanders}{Karlsruhe Institute of Technology, Germany}{sanders@kit.edu}{https://orcid.org/0000-0003-3330-9349}{}
\author{Stefan Walzer}{Cologne University, Germany}{walzer@cs.uni-koeln.de}{https://orcid.org/0000-0002-6477-0106}{}

\newcommand{\myauthorrunning}{Lehmann, Sanders, Walzer}
\authorrunning{\myauthorrunning}
\Copyright{\myauthorrunning}
\hypersetup{
  colorlinks=true,
  pdftitle={\mytitlerunning},
  pdfauthor={\myauthorrunning},
  pdfsubject={}
}


\begin{document}
%\linenumbers
%\openup 1em % Double spacing for easier note-taking
\maketitle

\begin{abstract}
We present {\bf Sli}ding Blo{\bf ck} Hashing
(Slick), a simple hash table data structure that
combines high performance with very good space
efficiency.  This preliminary report outlines
avenues for analysis and implementation that we
intend to pursue.
\end{abstract}
%\setcounter{page}{0}
% \paragraph*{Supplementary Material.}
% All implementations presented in this paper and scripts to reproduce our experimental evaluation are available under the GPLv3 license.
% PaCHash and Separator implementation: \url{https://github.com/ByteHamster/PaCHash}.
% Scripts for reproduction of results: \url{https://github.com/ByteHamster/PaCHash-Experiments}.

%\clearpage
%%%%%%%%%%%%%%%%%%%%%%%%%%%%%%%%%%%%%%%%%%%%%%%%%%%%%%%%%%%%%%%%%%%%%%
\section{Introduction}

Hash tables support the management of a set of
elements under search, insertion and deletion with
all these operations working in (expected)
constant time.  They are one of the most widely
used data structures and often important for the
performance of an application.

A painful tradeoff between space consumption and
speed of the used hash tables is one reason why so
many different variants are being considered. We
introduce {\bf Sli}ding blo{\bf ck} hashing
(Slick) which makes this tradeoff much less
punishing by allowing high performance even with
very small space overhead. This works both in
theory and practice \hpfrage{Not yet} and can be achieved already
with quite simple implementations.

Slick mostly adheres to the successful approach of
\emph{closed hashing} where elements are directly
stored in a table $T$. Perhaps the most widely
used closed hashing approach is \emph{linear
probing} \cite{Pet57,Knu98,SMDD19} which inserts elements by hashing them to
a table entry and then scans linearly for a free
slot.  Linear probing can be made highly space
efficient but only at the price of high search
times. Suppose the table has $m=(1+\epsilon)n$
slots available to store a set $S\subseteq E$ of $n$ elements.
Then the expected insertion time (and search time
for elements not in $S$) is $\Om{1/\epsilon^2}$ because
long clusters of full table entries are formed.

Slick defuses this problem by allowing some
elements to be \emph{bumped} to a next layer of
the data structure (often called a
\emph{backyard}) \cite{broder1990multilevel,FPSS05,arbitman2010backyard,bender2021all,pandey2022iceberght}.
By \emph{overloading} the table, i.e., choosing $m<n$,
this helps to fill it almost completely while only bumping
a small fraction of the elements.
Slick also stores metadata that
maps keys to a narrow range of table entries that
can possibly contain an element with that key. Thus, worst
case constant query time can be achieved.  It
turns out that the overhead for storing metadata
with bumping and location information is very
small.  Indeed, we conjecture that a variant of
Slick allows succinct storage of $S$ with space
$\log\binomial{|E|}{n}(1+o(1))$ while maintaining
constant time for all basic operations -- see
\cref{ss:succinct}.

After introducing some preliminaries in
\cref{s:prelim}, we introduce basic Slick in
\cref{s:basic}. \Cref{s:advanced} discusses
several variants and advanced features, including
a succinct representation. After discussing
related work in \cref{s:related}, we conclude in
\cref{s:conclusion} with some possible avenues for
future work.

%%%%%%%%%%%%%%%%%%%%%%%%%%%%%%%%%%%%%%%%%%%%%%%%%%%%%%%%%%%%%%%%%%%%%%
\section{Preliminaries}\label{s:prelim}

A hash table stores a set $S\subseteq E = K \times V$ of
$n=|S|$ key-value pairs for arbitrary universes $K$ and $V$. The pairs are also called elements.
Every key may only appear in one element, i.e.\ $S$ is a functional relation.
%
% A hash table stores a set $S\subseteq E$ of
% $n=|S|$ elements.  In this paper we assume that
% elements are key-value pairs, i.e., $E=K\times V$ for arbitrary universes $K$ and $V$.
Closed hashing stores these elements in an array
$T[0..m-1]$ providing space for $m$
elements.\footnote{In this paper, $a..b$ is a
shorthand for $\set{a,\ldots,b}$.}  A hash
function $h$ applied to the keys helps finding
these elements. In the following, we assume that
$h$ behaves like a truly random hash function.

Our model of computation is the standard RAM model
\cite{Sheperson-Sturgis} allowing constant time
operations on operands of size $\Oh{\log n}$ where
$n$ is the input size (see
e.g. \cite[Section~2.2]{SMDD19short} for details).

%%%%%%%%%%%%%%%%%%%%%%%%%%%%%%%%%%%%%%%%%%%%%%%%%%%%%%%%%%%%%%%%%%%%%%
\section{Basic Sliding Block Hashing}\label{s:basic}
In this section, we first introduce the basic data structure and the find operation in \cref{s:find}.
Then, we explain the insertion operation in \cref{s:insertion} and the bulk construction in \cref{s:bulkConstruction}.
Finally, we give details on the deletion operation in \cref{ss:deletion}.

%----------------------------------------------------------------------
\subsection{The Data Structure and Operation \Id{find}}\label{s:find}

The basic idea behind Slick is very simple.  We
try to store most elements in table $T$ as in
closed hashing. The main hash function $h$ maps
elements to the range $0..m/B-1$, i.e., to
\emph{blocks} for which $T$ has an average
capacity of $B$ available. Ideally (and
unrealistically), block $b_i$ would contain up to
$B$ elements and it would be stored in table entries
$T[iB..iB+B-1]$. Slick makes this realistic by
storing \emph{metadata} that indicates the
deviation from this ideal situation.

How to implement this precisely, opens a large
design space. We now describe a simple solution
with a number of tuning parameters.  The elements
of $S$ mapped to block $b_i$ are stored
contiguously in a range of table entries starting
at position $iB+o_i$ where $o_i$ is the
\emph{offset} of $b_i$ -- blocks are allowed to
\emph{slide}. After this block, there may be a
\emph{gap} of size $g_i$
%% \psfrage{discuss: I am not
%%   sure what range the gap size can
%%   have. Initially, in an empty table, all gaps
%%   have size $B$. I guess we can maintain the
%%   invariant that this is never exceeded. However,
%%   when bumping a block of size (up to) $\maxB$, it
%%   can happen that all its elements are
%%   bumped. Resulting in a gap of this size. In that
%%   case, we can slide the right neighbor to the
%%   left, possibly recursively.}
of unused table
cells before the next block starts. Metadata
explicitly stores the gap size.%
\footnote{This can be implemented using very
little additional space: We only need a single
code-word for the metadata of a block to indicate
that this block has a nonzero gap. In that case,
we have an entire unused table cell available to
store the gap size and the remaining metadata for
that block.  In particular, if we have $k$-bit
thresholds, we can use $2^k+1$ threshold values
and set $\maxOffset=2^k-2$. We then have
$|0..\maxThreshold\times 0..\maxOffset|=2^{2k}-1$
leaving one code word for encoding a nonempty gap.\psfrage{shortened footnote}}
%% For example, suppose we
%% have 64-bit elements and $B=8$.  We can choose
%% $\maxThreshold=16$ to allow 4-bit thresholds,
%% i.e., there are 17 possible threshold values. If
%% we then choose $\maxOffset=14$, we need
%% $17\cdot15=255$ code-words to encode offset and
%% threshold metadata. Using this, leaves one
%% code-word (e.g., 255) to encode a nonzero gap. The
%% overhead for metadata is then about $8/(8\cdot
%% 64)\approx 1.6$ \% of the data for storing
%% elements. More generally, if we have $k$-bit
%% thresholds, we can use $2^k+1$ threshold values
%% and set $\maxOffset=2^k-2$. We then have
%% $|0..\maxThreshold\times 0..\maxOffset|=2^{2k}-1$
%% leaving one code word for encoding a nonempty
%% gap.\hpfrage{Move to succinct
%%   section?}\psfrage{rather not. The succinct
%%   section is rather theoretical and would not care
%%   about the small constant factor in metadata
%%   saved here. I want it here to illustrate that it
%%   is not wasteful to encode gaps explicitly rather
%%   than having special empty elements. In a first
%%   implementation I would not implement this
%%   refinement. But perhaps encapsulate metadata in
%%   a way that allows to add that later. At some
%%   point this could move to implementation
%%   details.} \hpfrage{Maybe move it to a paragraph?
%%   This huge footnote looks a bit weird to
%%   me}\psfrage{Yes looks weird. But I want a simple
%%   story for people who skip the footnote at first
%%   reading. But I have shortened the footnote}}
This has the added benefit, that, in contrast to
previous closed hashing schemes, there is no need
to explicitly represent an empty element.
\begin{figure}[t]
  \centering
  \includegraphics[scale=0.9]{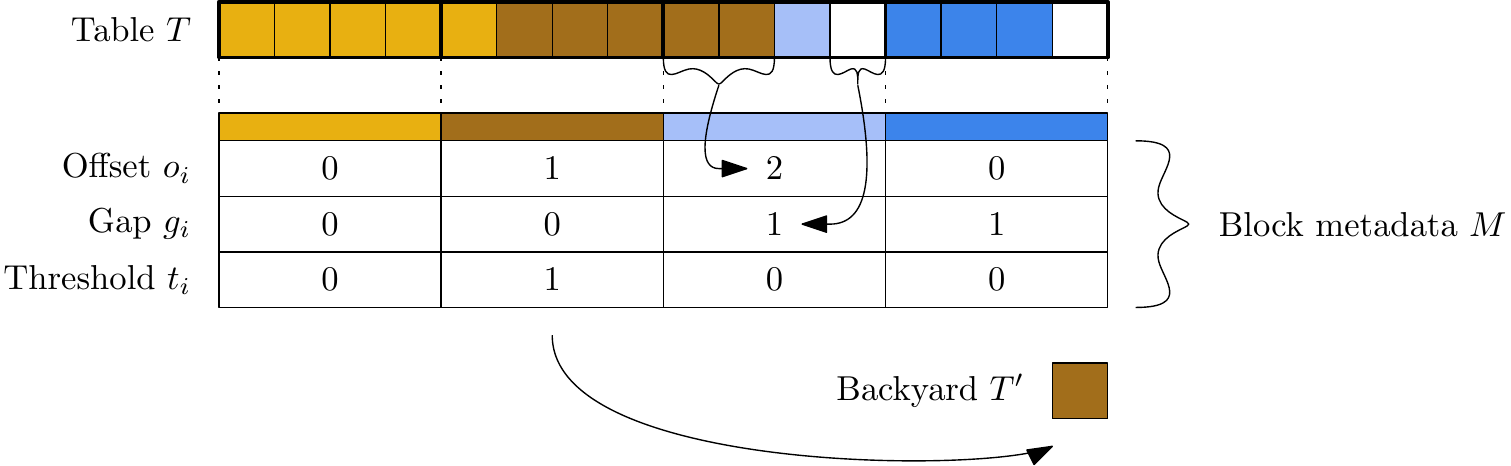}
  \caption{\label{fig:illustration} Illustration
    of the data structure. Input objects are
    colored by their hash function value. $B=4,
    \hat B=5, m=16, n=15$.}
\end{figure}

Metadata should use very little space.
We therefore limit the maximum offset to a value $\maxOffset$.
We also want to support fast search and therefore limit the block size
by parameter $\maxB$. With these constraints on position and size of blocks,
it may not be possible to store all elements of the input set $S$ in table $T$.
We therefore allow some  elements to be
\emph{bumped} to a \emph{backyard} $\backyard$.%
\footnote{The term ``bumped'' comes from the practice
of airlines to bump passengers if overbooking
leads to doubly sold seats.}  The backyard can be
any hash table. Since, hopefully, few elements
will be bumped, space and time requirements of the
backyard are of secondary concern for now. We
adopt the approach from the BuRR retrieval data
structure \cite{DHSW22} to base bumping decisions
on \emph{thresholds}: A threshold hash function
$\threshold(k)$ maps keys to the range
$0..\maxThreshold-1$. Metadata stores a threshold $t_i\in 0..\maxThreshold$
for block $b_i$ such that elements with
$\threshold(k)<t_i$ are bumped to $T'$.  We also use
the observation from BuRR that \emph{overloading}
the table, i.e., choosing $m<n$ helps to arrive at
tables with very few empty cells.
\hpfrage{Mention bumping of separator hashing here already?}

The pseudocode in \cref{alg:slick} summarizes a
possible representation of the above scheme and
gives the resulting search \hpfrage{search vs locate vs find} operation.
\Cref{fig:illustration} illustrates the data structure.
The
metadata array $M$ contains an additional slot
$M[m/B]$ with a sentinel element helping to ensure
that block ends can always be calculated and that
no elements outside $T$ are ever
accessed.\footnote{We could also wrap these blocks
back to the beginning of $T$ as in other closed
hashing schemes or extend $T$ to accommodate slid blocks.} This implementation has tuning
parameters $m$, $B$, $\maxB$, $\maxOffset$, and
$\maxThreshold$.  We expect that values in
$\Th{B}$ will be good choices for $\maxB$,
$\maxOffset$, and $\maxThreshold$.  Concretely, we
could for example choose $\maxB=2B$ and
$\maxOffset=\maxThreshold=B$.  This leads to space
overhead of $\Oh{\log B}$ bits for the metadata of
each block which can be amortized over $B$
elements.

%\hpfrage{Is it still possible to encode the gap
%  values by a negative offset value in the next
%  block? I didn't see this in the explanation
%  yet.}\psfrage{no. I have abandoned that
%  representation for now as it restricts where
%  gaps can be and makes deletions more complicated
%  and more expensive. With the trick for
%  representing gaps in the footnote, it also does not save
%  metadata space. Stefan once said that that he likes
%  that representation because it kind of
%  ``normalizes'' the data structure. We have to
%  see whether this helps enough with the analysis
%  to resurrect that representation.}

\begin{figure}
  \begin{code}
    \Class \Id{SlickHash}($m, B, \maxB, \maxOffset, \maxThreshold:\nplus$, $h:E\rightarrow 0..m/B-1$)\+\\
      \Class \Id{MetaData} = $\overbrace{o:0..\maxOffset}^{\text{offset}}\times \overbrace{g:0..\maxB}^{\text{gap}} \times \overbrace{t:0..\maxThreshold}^{\text{threshold}}$\\[2mm]
      $T$ : \Array$[0..m-1]$ \Of $E$\RRem{main table}\\
      \DeclareInit{$\meta$}{\Array$[0..m/B]$ \Of \Id{MetaData}}{$(0,B,0)^{m/B}\circ(0,0,0)$}\\ 
      \backyard : \Id{HashTable}\RRem{backyard}\\[2mm]

      \Function \Id{blockStart}($i$: $\nat$) \Return $Bi+o_i$\\
      \Function \Id{blockEnd}($i$: $\nat$) \Return $Bi+B+o_{i+1}-g_i-1$\\
      \Function \Id{blockRange}($i$: $\nat$) \Return $\Id{blockStart}(i)..\Id{blockEnd}(i)$\\[2mm]
      
      \Rem{locate an element with key $k$ and return a reference in $e$}\\
      \Funct{\Id{find}}{$k$: $K$, $e$: $E$}{\Boolean}\+\\
        $i\Is h(k)$\\
        \If $\threshold(k)< t_i$ \Then \Return \backyard.\Id{find}$(k, e)$\RRem{bumped?}\\
        \If $\exists j \in\Id{blockRange}(h(k))\gilt \Id{key}(T[j])=k$ \Then
           $e\Is T[j]$;\quad \Return \True\RRem{found}\\
        \Else \Return \False\RRem{not found}
  \end{code}
  \caption{\label{alg:slick}Pseudocode for a
    simple representation of Slick and the
    operation \Id{find}.  \hpfrage{I think the
      outer class specification can be
      removed. Especially the definition of
      MetaData is not really used as it is
      defined. Instead, $o_i,g_i,t_i$ are accessed
      directly. This just makes it longer and
      therefore more ``scary'' to
      read.}\psfrage{Perhaps. I tend to keep it at
      least for the tr as it fills in some gaps
      like ranges of variables, initialization
      etc. The class declration itself collects
      the set of tuning paramters in one place.}
    \hpfrage{The code feels quite
      low-level. Instead of setting a function
      parameter to a reference of an object and
      returning true/false, I would just return
      the object itself in the pseudocode (or
      $\bot$). This makes it easier to
      understand.}\psfrage{I actually had
      that. However this might be confusing as as
      in the text I proudly explain that we do not
      need a special element $\bot$. We could have
      a return type $\bot\cup E$ but one would
      usually not implement it that way in C++
      (but in Rust?). }\hpfrage{For pseudocode I would not care too much if this is close to the implementation in C++ and instead focus on easy readability. The variant that can return $\bot$ is easier to read.}  }
\end{figure}

%----------------------------------------------------------------------
\subsection{Insertion}\label{s:insertion}

\begin{figure}
  \begin{code}
      \Procedure \Id{insert}($e$: $E$)\RRem{insert element $e$, noop if already present}\+\\
        $k\Is \Id{key}(e)$;\quad $i\Is h(k)$\RRem{current block}\\
        \If $\threshold(k)< t_i$ \Then \backyard.\Id{insert}$(e)$;\quad \Return\RRem{$e$ is already bumped}\\
        \If $\exists j \in\Id{blockRange}(h(k))\gilt \Id{key}(T[j])=k$ \Then \Return\RRem{already present}\\
        \If $\overbrace{|\Id{blockRange}(i)|=\maxB}^{\text{block too large}}$ \Or $\overbrace{\Not(g_i>0\,||\,\Id{slideGapFromLeft}(i) \,||\, \Id{slideGapFromRight}(i))}^{\text{no empty slot usable}}$ \Then\+\\
          \Rem{bump $e$ or some element from block $b_i$}\\
          $t'\Is 1+\min\setGilt{\delta(x)}{x\in \set{k}\cup\setGilt{\Id{key}(T[j])}{j\in \Id{blockRange}(i)}}$\\
          $t_i\Is t'$\\
          $j\Is \Id{blockStart}(i)$\\
          \While $j\leq \Id{blockEnd}(i)$ \Do\+\RRem{Scan existing elements. Bump them as necessary}\\
            \If $\threshold(\Id{key}(T[j]))<t'$ \Then\+\\
              $\backyard.\Id{insert}(T[j])$\RRem{move to backyard}\\
              $T[j]\Is T[\Id{blockEnd}(i)]$;\quad
              $g_i\Increment$\-\RRem{remove from $T$}\\
            \Else $j\Increment$\-\\
          \If $\threshold(k)< t'$ \Then $\backyard.\Id{insert}(e)$; \quad\Return\-\\
        $g_i\Decrement$;\quad
        $T[\Id{blockEnd}(i)]\Is e$\RRem{insert $e$ into an unused slot}\\
        \Return\-\\[2mm]  

      \Rem{Look for a free slot to the right and move it to block $b_i$ if successful}\\
      \Funct{\Id{slideGapFromRight}}{$i_0$: $\nat$}{\Boolean}\+\\
        $i \Is i_0$\\
        \While $g_i = 0$ \Do\RRem{look for a free slot}\+\\
          \If $i\geq m/B\vee o_i=\maxOffset$ \Then\Return \False\RRem{further sliding right is impossible}\\
          $i\Increment$\-\\
        $g_i\Decrement$\\
        \While $i>i_0$ \Do\+\RRem{shift free slot towards block $b_{i_0}$}\\
          \Rem{Slide $b_i$ to the right}\\
          $T[\Id{blockEnd}(i)+1]\Is T[\Id{blockStart}(i)]$\\
          $o_i\Increment$\\
          $i\Decrement$\-\\
        $g_i\Increment$\\
        \Return \True
  \end{code}
  \caption{\label{alg:insert}Pseudocode for insertion into Slick
    Hash Tables. Function \Id{slideGapFromLeft}
    works similarly to function
    \Id{slideGapFromRight} -- it scans left looking
    for a free slot, failing if it encounters an
    offset of 0. It slides blocks left by moving
    their last element to the last free slot of
    the previous block thus decreasing their offset.}
\end{figure}
\begin{figure}
  \centering
  \includegraphics[scale=0.9]{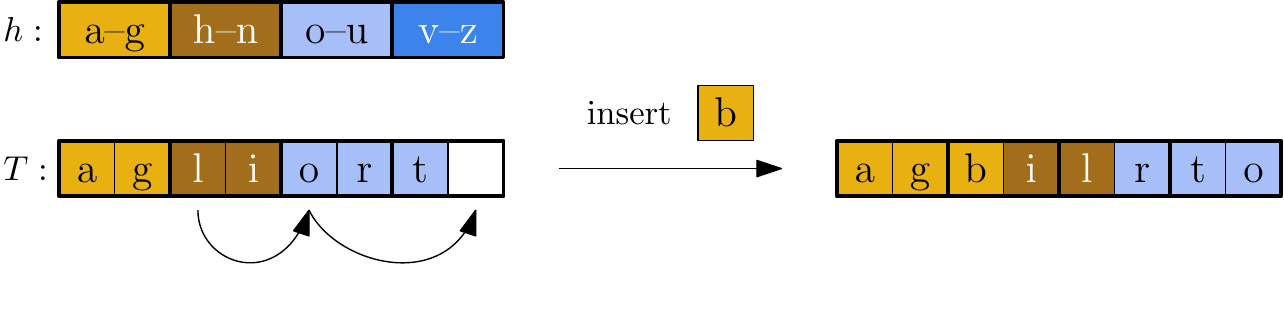}
  \caption{\label{fig:insert}Example for inserting
    $b$ into a Slick hash table with $B=2$ that
    previously stored
    $S=\set{a,l,g,o,r,i,t,h,m}$. }
\end{figure}

To insert an element $e$ with key $k$ into a Slick
hash table, we 
first find the
block $i=h(k)$ to which $e$ is mapped.  Then we check
whether the threshold for block $i$ implies
that $e$ is bumped. In this case, $e$ must be
inserted into the backyard $\backyard$.
Otherwise, the data structure invariants give us
quite some flexibility how to insert an element.
This implies some case distinctions but also allows
quite efficient insertions even when the
table is already almost full.

A natural goal is to
insert $e$ into block $b_i$ if this is possible
without bumping other elements.  The pseudocode in
\cref{alg:insert} describes one way to do this.
This goal is unachievable
when $b_i$ already contains the maximal number
of $\maxB$ elements. In that case, we must bump
some elements to make room for $e$.\footnote{One
could think that the most simple solution it to
bump $e$ itself. However, this might incur
considerable ``collateral damage'' by additionally
bumping all elements from $b_i$ with smaller
or equal threshold.} Once more there are many ways to achieve this.
We describe a fast and simple variant.%
\footnote{This variant has the disadvantage that
it sometimes bumps several elements. A more
sophisticated yet more expensive variant could
look for neighboring blocks where bumping a
single element is possible.}  We look for the
smallest increase in the threshold of $b_i$
that bumps at least one element from that block
(including the new element $e$ itself).  We set
the threshold accordingly and bump the elements
implied by this change.  Now, either $e$ is bumped
and we are done with the insertion or a free slot
after $b_i$ is available.

If block $b_i$ is not filled to capacity $\maxB$, we
try to insert $e$ there.  However, this may not be
directly possible because the gap behind $b_i$
could be empty. In that case we can try to slide
neighboring blocks to open such a gap.
Algorithm~\ref{alg:insert} does that by first
trying to slide $b_i$ and some of its left
neighbors by one position to the left. This may
fail because before finding a nonempty gap, a
block with offset 0 may be encountered that cannot
be slid to the left.%
\footnote{Indeed, this will always fail when we
use no deletions. Hence, the attempt to slide left
can be omitted in that situation.}  If sliding left failed,
function \Id{slideGapFromRight} tries to slide
blocks to the right of $i$ by one position to the
right. Once more, this may fail because blocks
that already have maximum offset cannot be slid
to the right. If there is a range of slidable
blocks starting a $b_{i+1}$ and ending at a
block followed by a nonempty gap, the actual
sliding can be done efficiently by moving only one
element per block. \Cref{fig:insert} gives an example. The first element is appended
to the block, filling the first gap element and
growing the gap of the previous block.

If neither sliding left nor sliding right can open
a gap after block $b_i$, the same bumping procedure
described for full blocks is used.  If sliding was
successful (including the case of an empty
right-slide for the case that $b_i$ already
had a nonempty gap), the gap after $b_i$ is
nonempty and element $e$ can be appended to $b_i$.

\begin{corollary}\label{cor:insertTime}
  A call of insert takes time
  $\Ohsmall{\maxB+s}+T_{\mathrm{bump}}$ where $s$
  is the number of considered blocks and where
  $T_{\mathrm{bump}}$ is the time incurred in the
  backyard.
\end{corollary}\psfrage{can we show that the expectation of
  $s$ is $\Oh{1+\frac{m}{B\cdot\Id{emptySlots}}}$,
  implying that the cost is $\Oh{B}$ when
  $\Id{emptySlots}=\Om{m/B^2}$? I am afraid though
  that at some point we also suffer from
  clustering as in linear probing. But this will
  set in later once blocks with empty slots get
  scarce. Sth like query time $1/(B^2(1-n/m)^2)$?}
%% \begin{corollary}
%%   If $\maxOffset\geq n$ and wrapping around is
%%   used, procedure slideGapFromRight will always
%%   succeed if there is any empty slot.  This
%%   implies that we can fill the table to its full
%%   capacity in that case.
%% \end{corollary}

The insertion routine described so far is quite
effective in filling the table, however, this gets
expensive when the table is almost
full. Therefore, in cases where this can happen,
one should use a less aggressive insertion routine
that bumps if there is no nearby empty slot.  What
``nearby'' means can be controlled with a tuning
parameter $\hat{s}$ that controls the maximal number of
blocks to slide. \Cref{cor:insertTime} implies that $\hat{s}=\Ohsmall{\maxB}$
might be a good choice.

We can also consider a more sophisticated
insertion routine supported by additional
metadata.  When routine \Id{insert} from
\cref{alg:insert} fails to find a gap to the left
or the right, it has identified a \emph{blocked
cluster (bluster)} of full blocks starting with a block with
offset $0$ and ending with a block with offset
$\maxOffset$\psfrage{More precisely, we are looking at maximal such clusters. An insertion already terminates at the first block with maximal overlap it encounters. The bluster will in general extend further to the right.}. When the table is almost full,
blusters can be large and they can persist
over many insertions.%
\footnote{In the absence of deletions, blusters can
only be broken when bumping happens to bump at
least two elements.}
Thus, it might help to mark blocks inside blusters.
An insertion into a bluster can then immediately bump.
If blocking flags are represented by a separate bit array, they
can be updated in a bit parallel way.

%----------------------------------------------------------------------
\subsection{Bulk Construction}\label{s:bulkConstruction}
\begin{figure}
  \begin{code}
    \Procedure \Id{greedyBuild}($S$: \Id{Sequence} \Of $E$)\RRem{build a Slick hash table from $S$}\+\\
      $\Id{bumped}\Is\seq{}$\RRem{bumped elements}\\
      sort $S$ lexicographycally by $(h(e),\threshold(e))$\\
      $o\Is 0$\RRem{offset}\\
      \ForFromTo{i}{0}{m/B-1}\RRem{for each block}\+\\
%        $b\Is\seq{}$;\quad \While $h(S.\Id{front})=i$ \Do $b.\Id{pushBack}(S.\Id{popFront})$\RRem{extract block $b_i$ from S}\\
        $b\Is\seqGilt{e\in S}{h(\Id{key}(e))=i}$\RRem{extract block $b_i$ from S}\\
      $t\Is 0$\RRem{threshold for $b_i$}\\[-3mm]
        \Id{excess}\Is$\max(\overbrace{|b|-\maxB}^{|b|\leq\maxB}, \overbrace{o+iB+|b|-m}^{|T|=m}, \overbrace{o+|b|-B-\maxOffset}^{o_{i+1}\leq\maxOffset})$\\
        \If $\Id{excess}>0$ \Then\+\\
          \ForFromTo{j}{1}{\Id{excess}} $\Id{bumped}.\Id{pushBack}(b.\Id{popFront})$\RRem{bump to fit}\\
          $t\Is \threshold(\Id{bumped}.\Id{last})+1$\RRem{adapt threshold}\-\\
        \While $|b| > 0 \wedge \delta(b.\Id{front}) < t$ \Do\+\\
          $\Id{bumped}.\Id{pushBack}(b.\Id{popFront})$\-\\
        $M[i]\Is (o,\max(0,B-o-|b|),t)$\RRem{write metadata for $b_i$}\\
        \ForFromTo{j}{0}{|b|-1} $T[iB+o+j]\Is b[j]$\RRem{write $b_i$ to $T$}\\
        $o\Is\max(0, o+|b|-B)$\RRem{next offset}\-\\
      $M[m/B] \Is (0,0,0)$\RRem{sentinel metadata}\\
      $\backyard.\Id{build}(\Id{bumped})$
  \end{code}
  \caption{\label{alg:build}Pseudocode for greedy bulk construction of Slick hash tables. }
\end{figure}

\Cref{alg:build} gives pseudocode for a simple
greedy algorithm for building a Slick hash table.
The elements are sorted by block (and threshold)
and then processed block by block.  The algorithm
tries to fit as many elements into the block as
permitted by the constraints that a block must
contain at most $\maxB$ elements, that the
offset must not exceed $\maxOffset$, and that the last available table cell is $T[m-1]$.  Violations
of these constraint are repaired by bumping a
minimal number of elements from the current block.

\begin{theorem}
  Construction of a Slick hash table using
  \Id{greedyBuild} can be implemented to run in
  deterministic time $\Oh{|S|}$ plus the time for
  constructing the backyard. As an external memory algorithm,
  greedy build has the same I/O complexity as sorting.
\end{theorem}
\begin{proof}
  Using LSD-radix-sort (e.g., see \cite[Section~5.10]{SMDD19short}),
  $S$ can be sorted in linear time.%
  \footnote{Indeed, for $\maxThreshold\in\Oh{B}$
  or in expectation, it suffices to use plain
  bucket sort using $\Oh{n}$ buckets.}  All other
  parts of the algorithm are simple to implement
  in linear time.

  For the external memory variant, we observe that
  sorting the input can pipeline its sorted output
  into the construction process, requiring
  internal memory only for one block at a time.
\end{proof}

\Id{greedyBuild} seems to be a good heuristic for bumping
few elements. We are also optimistic that we can
prove that using slight overloading, the number of
empty cells can be made very small
($me^{-\Omsmall{\maxOffset}}+o(m)$ assuming $\maxOffset = \Oh{B}$ using an
analysis similar to \cite{DHSW22}). Moreover, for
fixed values of the tuning parameters, $B$,
$\maxB$, and $\maxOffset$, we can derive the
expected number of empty cells for large $n$ using
Markov chains. This allows us to study the impact
of different loads $n/m$
analytically.\psfrage{later: describe that Markov
  chain in an appendix.}

However, \Id{greedyBuild} is not optimal.  For
example, suppose we have $B=\maxOffset=4$, $o_2=3$,
$|b_2|=6$, and all elements in block $b_2$ have the same
threshold value. Then we have to bump all elements
from $b_2$ because otherwise
$o_{3}=3+6-4=5>4$ violating the constraint on
the offset of $b_3$.  This leaves an empty
cell after $b_2$.  However, it might have been
possible to bump one additional element from $b_1$ which would have resulted in $o_2=2$ so that no elements would have to be
bumped from $b_2$.  We consider to try heuristics
that avoid empty cells in a block when they
arise, going backwards to bump elements from
previous blocks.  We can also compute an optimal
placement in time $\Oh{n\maxOffset}$ using dynamic
programming%
\footnote{Roughly, we consider the blocks from
left to right and compute, for each possible
offset value $o$, the least number of elements
that need to be bumped to achieve an offset of at
most $o$.}.
% \footnote{Roughly, we
% compute for each possible offset value and each
% possible bumping decision what would be the total
% number of empty cells and the offset for the next
% block.}

%----------------------------------------------------------------------
\subsection{Deletion and Backyard Cleaning}\label{ss:deletion}

\begin{figure}
  \begin{code}
    \Procedure \Id{delete}($k$: $K$)\RRem{delete the element with key $k$, noop if not present}\+\\
      \If $\threshold(k)<t_i$ \Then $\backyard.\Id{delete}(k)$;\quad\Return\RRem{element can only be bumped}\\
      \If $\exists j \in\Id{blockRange}(h(k))\gilt \Id{key}(T[j])=k$ \Then\+\RRem{found}\\
        $T[j]\Is T[\Id{blockEnd}(h(k))]$\RRem{overwrite deleted element}\\
        $g_{h(k)}\Increment$\RRem{extend gap}
  \end{code}
  \caption{\label{alg:delete}Pseudocode for deletion from Slick
    Hash Tables. }
\end{figure}

Deleting an element is almost as simple as finding
it. We just overwrite it with the last element of
the block and then increment the gap
size. \Cref{alg:delete} gives pseudocode.

A problem that we have to tackle in applications
with many insertions and deletions is that the
routines presented so far bump elements but never
unbump them.  \psfrage{discuss later: iceberg can use
  an unmanaged backyard. I think this will
  not work for Slick as threshold remember
  previously bumped elements even when they have
  been deleted. Discuss that somewhere?}  In a situation with basically stable
$|S|$, the backyard $\backyard$ will therefore keep
growing.  This effect can be countered by
\emph{backyard cleaning}. When there are enough
empty cells in the primary table $T$ to
accommodate $\backyard$, we can reinsert all
elements from $\backyard$ into $T$.  For each
block $b_i$, its threshold can be reset to $0$
unless insertion for a backyard element causes
bumping. An important observation here is that we
never have to bump an element that was in $T$
before backyard cleaning and that we do not have
to look at threshold values of those
elements. This implies that the $\Ohsmall{\maxB}$
term from \cref{cor:insertTime} can be replaced by
the number of backyard elements inserted into the
affected block.  We also expect that we can design
a cleaning operation that works more efficiently
than inserting backyard elements one-by-one.  For
example, this can be done by sorting the backyard similar to the
build operation and then merging backyard and main
table in a single sweep\psfrage{later: design such
  a routine?}.

Now suppose the
backyard has size $\Oh{m/B}$.  Then backyard
cleaning can be implemented to run in expected time
$\Oh{m/B}$.  If we choose the table size in such a
way that $\Oh{m}$ insert or delete operations
cause only $\Oh{m/B}$ bumps, then the backyard and its
management will incur only a factor $\Oh{1/B}$
overhead in space or time.

%%%%%%%%%%%%%%%%%%%%%%%%%%%%%%%%%%%%%%%%%%%%%%%%%%%%%%%%%%%%%%%%%%%%%%
\section{Variants and Advanced Features}\label{s:advanced}

We begin with two variants of Slick that may be of
independent interest.  Linear cuckoo hashing
described in \cref{ss:luckoo} is almost a special
case of Slick that has advantages with respect to
the memory hierarchy. Bumped Robin Hood hashing from
\cref{ss:ribbonHood} is conceptually even simpler
than Slick and may allow even faster search at the
price of slower insertions.

\Cref{ss:succinct} is a key result of this section
showing how to further reduce space consumption up
to the point that the data structure gets
succinct. Using bit parallelism allows that while
maintaining constant operation
times. \Cref{ss:parallel} briefly discusses how
Slick can be parallelized.

%----------------------------------------------------------------------
\subsection{\underline{L}inear C\underline{uckoo} (Luckoo) Hashing}\label{ss:luckoo}

{\bf L}inear C{\bf uckoo} (Luckoo) Hashing is closely related to Slick hashing but
more rigidly binds the elements to blocks.
Luckoo hashing subdivides the table $T$ into blocks of size $B$ and 
maintains the invariant that each unbumped element $e\in S$ is either
stored in block $h(e)$ or in block $h(e)+1$.%
\footnote{A generalization could look at $k$ consecutive blocks.}
This can be implemented as a special case of Slick hashing with $\maxOffset=B$, $\maxB=2B$,
and the additional constraint that 
$o_i+|b_i|\leq 2B$. The main advantage over general Slick hashing is that we can now
profit from interleaving metadata with table entries in physical memory blocks of the machine.
This way, a find-operation incurs at most two cache faults.%
\footnote{Note that hardware prefetchers may help to
execute two contiguous memory accesses more
efficiently than two random ones.}
Another potential advantage is that storage of offset and gap metadata is now optional as
the data structure invariant already defines which $2B$ table cells contain a sought element.
This might help with SIMD-parallel implementations.

Luckoo is also useful in the context of truly
external memory hash table, e.g., when used for
hard disks or SSDs. We obtain a dynamic hash table
that is able to almost completely fill the table
and where \Id{find} and \Id{delete} operations
access only 2 consecutive blocks. \hpfrage{Mention similarity with external cuckoo hashing where keys have all except 1 of their choices on the same page?}
Insertions look
at a consecutive range of blocks.  No internal
memory metadata is needed.  Most external memory
hash tables strive to support operations that look
at only a single physical block most of the time.
We can approximate that by subdividing a physical
block into $k$ Slick blocks.  That way, \Id{find}
will only access $1+\frac{1}{k}$ physical blocks
in expectation.%
\footnote{For example, let us consider the case of
an SSD with physical blocks of size 4096, $k=8$,
$B=8$, and 63-bit elements. Then we have enough
space left for 8 bits of metadata per block. On
average, one in 8 \Id{find} operations will have
to access two physical blocks.}
%----------------------------------------------------------------------
\subsection{Nonbumped Slick -- Blocked Robin Hood Hashing (BloRoHo)}\label{ss:bloRoHo}

Robin Hood hashing \cite{RobinHood} is a variant
of linear probing that reduces the cost of
unsuccessful searches by keeping elements sorted
by their hash function value.  
Slick without bumping can be seen as a variant of Robin
Hood hashing, i.e., elements are sorted by $h(k)$
and metadata tells us exactly where these elements
are.  We get expected search time $\Oh{B}$ and
expected insertion time bounded by
$\Ohsmall{B+T_{\mathrm{RH}}/B}$ where
$T_{\mathrm{RH}}$ is the expected insertion time
of Robin Hood hashing. For  sufficiently filled tables
and not too large $B$, this is faster than basic
Robin Hood hashing. Since the bumped version of Slick
cannot be asymptotically slower, this also gives
us an upper bound on the insertion cost of general Slick.

Actually implementing BloRoHo saves the cost and
complications of bumping but pays with giving up
worst-case constant \Id{find}. It also has to take
care that metadata appropriately represents all
offsets which can get large in the worst case.
Rehashing or other special case treatments may be needed.

%----------------------------------------------------------------------
\subsection{Bumbed Robin Hood Hashing (BuRoHo)}\label{ss:ribbonHood}

We get another variant of Robin Hood hashing if
we start with basic Robin Hood hashing without offset or gap metadata but
allow bumping.
We can then enforce the invariant  that any
unbumped element $e\in S$ is stored in
$T[h(e)..h(e)+B-1]$ for a tuning parameter $B$.
As in linear probing, empty cells are indicated by
a special element $\bot$. Bumping information
could be stored in various ways but the most
simple way is to store one bit with every table cell $i$
whether elements $e$ with $h(\Id{key}(e))=i$ are bumped.

Compared to Slick, BuRoHo more directly
controls the range of possible table cells that
can contain an element and it obviates the need to
store offset and gap metadata. Searches will
likely be faster than in Slick. However,
insertions are considerably more expensive as we
have to move around many elements and evaluate
their hash functions. In contrast, Slick can skip
entire blocks in constant time using the available
metadata.

%----------------------------------------------------------------------
\subsection{Succinct Slick}\label{ss:succinct}

\subsubsection{Quotienting}\label{sss:quotient}

Cleary \cite{Cle84} describes a variant of linear
probing that infers $\log m$ bits of a key $k$
from $h(k)$. To adapt to displacements of
elements, this requires a constant number of
metadata bits per element. In Slick we can infer
$\log\frac{m}{B}$ key bits even more easily from
$h(k)$ since the metadata we store anyway already
tells us where elements of a block are stored.

To make this work, the keys are represented in a
randomly permuted way, i.e., rather than
representing a key $k$ directly, we represent
$\pi(k)$. We assume that $\pi$ and its inverse
$\pi^{-1}$ can be evaluated in constant time and
that $\pi:K\rightarrow 0..|K|-1$ behaves like a
random permutation. Using a chain of Feistel
permutations
\cite{LubRac88,NaoRei99,arbitman2010backyard},
this assumption is at least as realistic as the
the assumption that a hash function behaves like a
random function. Now it suffices to store
$k'=\pi(k)\bmod m/B$ in block
$\pi(k)\Div (m/B)$. To reconstruct a key $k$ stored as
$k'$ in block $i$, we compute
$k=\pi^{-1}(im/B+k')$.

With this optimization, the table entries now are stored succinctly except for
\begin{enumerate}
\item $\approx\log B$ bits
  per element that are lost because we only use bucket indices rather than table indices
  to infer information on the keys.
\item On top of this come $\Oh{\log(B)/B}$ bits per element of metadata.
\item Space lost due to empty cells in the table.
\item Space for the backyard.
\end{enumerate}
Items 1. and 2. can be hidden in $o(1)$ allowed for
succinct data structures. Items 3. and 4. become lower order terms
when $B\in \omega(1)$ \hpfrage{Should this be little-omega? Seems to be quite uncommon notation.}. Below, we will see how to do that while still
having constant time operations.

\subsubsection{Bit Parallelism}\label{sss:bitParallel}

With the randomly permuted storage of keys, any
subset of $f$ key bits can be used as a fingerprint
for the key.  If we take $\geq \log \maxB$ such
fingerprint bits, the expected number of
fingerprint collisions will be constant.
Moreover, for $f=\Ohsmall{\log \maxB}$, and
$\maxB=\Oh{\log(n)/\log\log n}$, we can use bit parallelism to process all 
fingerprints of a block in constant time.
To also allow bit parallel access to the right fingerprints, we have to store the
fingerprints separately from the remaining bits of the elements.%
\footnote{For a RAM model implementation of Slick,
it seems most simple to have a separate
fingerprint array that is manipulated analogously
to the element array. Since this can cause
additional cache faults in practice, we can also
use the Luckoo variant from \cref{ss:luckoo} with
a layout where a physical (sub)block contains
first metadata then fingerprints and finally the
remaining data for exactly $B$ elements. A
compromise implements general Slick with
separately stored metadata (hopefully fitting in
cache) but stores $B$ fingerprints and $B$
elements in each physical block.}

In algorithm theory, bit parallelism can do many
things by just arguing that lookup tables solve
the problem. However, this is often impractical
since these tables incur a lot of space overhead
and cache faults while processing only rather
small inputs. However, the operations needed for
bit parallel Slick are very simple and even
supported by SIMD-units of modern microprocessors.

Specifically, \Id{find} and \Id{delete} need to replicate the
fingerprint of the sought key $k$ and then do a
bit parallel comparison with the fingerprints of
block $h(k)$.  Operation \Id{insert} additionally
needs bit parallism for bumping.  By choosing
fingerprints large enough to fully encode the
thresholds $\delta(\cdot)$, we can determine the
elements with minimal fingerprint in a block using
appropriate vector-min and vector-compare
operations. The elements with minimal fingerprint
then have to be bumped one at a time which is
possible in constant expected time as the expected
number of minima will be a constant for $f\geq\log \maxB$.

We can now extend \cref{cor:insertTime}:
\begin{corollary}
  For $\maxB=\Oh{\log(n)/\log\log n}$ and $\log |E|=\Oh{\log n}$, operations
  \Id{find} and \Id{delete} can be
  implemented to work in constant expected time.
  The same holds for operation $insert$ if a variant is
  chosen that slides a constant number of blocks in expectation.
\end{corollary}

We believe that we can show that constant time
operations can be maintained, for example by choosing
$m=n+\Th{n/B}$ when in expectation $\Oh{n/B}$ elements will
be bumped. In this situation, we can afford to
choose a nonsuccinct representation for the
backyard.
Overall, for $B=\Oh{\log(n)/\log\log n}$ (and $\maxB=2B$, $\maxOffset=\Th{B}$, $\maxThreshold>\log B$) we will get
$$\log\binomial{|E|}{n}+\Oh{n\left(
  \overbrace{\log B}^{1.} +
  \overbrace{\frac{\log B}{B}}^{2.} +
  \overbrace{\frac{\log n}{B}}^{3.} +
  \overbrace{\frac{\log n}{B}}^{4.}\right)}=\log\binomial{|E|}{n}+\Oh{n(\log\log
  n)}$$ bits of space consumption.
Deletion and backyard cleaning can be done as described in \cref{ss:deletion}.

Even lower space consumption seems possible if we
overload the table (perhaps $m=n-\Th{B}$) and use
a succinct table for the backyard (perhaps using
non-overloaded Slick).  To maintain constant
expected insertion time, we can scan the metadata
in a bitparallel way in functions
\Id{slideGapFromLeft/Right}.  In addition, we
limit the search radius to $\Oh{B}$ blocks.  Only if we
are successful, we incur a nonconstant cost of
$\Oh{B}$ for actually sliding blocks.  Now
consider an insert-only scenario. The first
$n-\Oh{n/B}$ elements can be inserted in constant
expected time as in the non-overloaded
scenario. The remaining $\Oh{n/B}$ elements will
incur insertion cost $\Oh{B}$, i.e., $\Oh{n}$ in
total.\psfrage{perhaps even better using a more
  detailed analysis of expected insertion time as
  the table fills up.}\psfrage{what todo if we
  allow deletions? Perhaps store an table of free
  slots sorted by the hash function used for
  $\backyard$. After $\Oh{n/B^2}$ deletions, scan
  the backyard in some data parallel way for
  elements that can move there?}

%----------------------------------------------------------------------
%\subsection{Adaptive Growing}\label{ss:gowing}

% use extendable array
%% Rajeev Raman and Satti Srinivasa Rao. Succinct
%% dynamic dictionaries and trees. In Inter- national
%% Colloquium on Automata, Languages, and
%% Programming, pages 357–368. Springer, 2003.
% in practice use virtual memory
% MaiSan17, SanWas11

%----------------------------------------------------------------------
\subsection{Parallel Processing}\label{ss:parallel}

Many \Id{find} operations can concurrently access
a Slick hash table.  Operations \Id{insert} and
\Id{delete} require some kind of locking as do
many other hash table data structures (but,
notably, not linear probing; e.g., see
\cite[Section~4.6]{SMDD19short}). Often locking is
implemented by subdividing the table into an array
of segments that are controlled by one lock
variable.

We can parallelize operation \Id{build} and bulk
insert/delete operations by
subdividing the table into segments so that different threads
performs operations on separate segments of the
table. A simple implementation could enforce independent segments by
bumping data that would otherwise be slid across segment boundaries.
A more sophisticated implementation could use a postprocessing phase that
avoids some of this bumping (perhaps
once more in parallel).

%----------------------------------------------------------------------
%\subsection{Compressed Single Shot Bloom Filter AMQs}\label{ss:AMQ}

%%%%%%%%%%%%%%%%%%%%%%%%%%%%%%%%%%%%%%%%%%%%%%%%%%%%%%%%%%%%%%%%%%%%%%
\section{More Related Work}\label{s:related}

There is a huge amount of work on hash tables. We
do not attempt to give a complete overview but
concentrate on approaches that are direct
competitors or have overlapping ideas.

Compared to linear probing
\cite{Pet57,Knu98,SMDD19}, Slick promises faster
operations when the table is almost full.  In
particular, unsuccessful search, insertion and
deletion should profit.  Search and delete also
offer worst case deterministic performance
guarantees when an appropriate backyard is used in
contrast to expected bounds for linear probing%
\footnote{There are high-probability logarithmic bounds
for linear probing that require
fairly strong hash functions though
\cite{thorup2013simple}.}  An advantage for robust
library implementations of hashing is that Slick
does not require a special empty element.
Advantages of linear probing over Slick are
simpler implementation, lockfree concurrent
operation, and likely higher speed when ample
space is available.

\emph{Robin Hood hashing} \cite{RobinHood} improves
unsuccessful search performance of linear probing
at the expense of slower insertion. Slick and our
bumped Robin Hood variant described in
\cref{ss:ribbonHood} go one step further --
bumping obviates the need to scan large clusters
of elements during search and metadata allows
skipping them block-wise during insertion.

\emph{Hopscotch hashing}~\cite{HopscotchHashing}
and related approaches \cite{BThesisPerCellData}
store per-cell metadata to accelerate search of
linear probing. We see Slick as an improvement
over these techniques as it stores much less
metadata with better effect -- once more because,
thanks to bumping, clusters are not only managed
but their negative effect on search performance is
removed.

\emph{Cuckoo hashing}
\cite{CuckooHashing,FPSS05,DieWei07,maier2019dynamic}
is a simple and elegant way to achieve very high
space efficiency and worst case constant
\Id{find}-operations. Its governing invariant is
that an element must be located in one of two (or
more) blocks determined by individual hash
functions. Slick and linear cuckoo hashing
(Luckoo) described in \cref{ss:luckoo} achieve a
similar effect by mostly only accessing a single
range of table cells and thus improve locality and
allow faster insertion.

\emph{Bumping and backyards} have been used in
many forms previously.  \emph{Multilevel
  adaptive hashing} \cite{broder1990multilevel}
and \emph{filter hashing} \cite{FPSS05} bump
elements along a hierarchy of shrinking
layers. These approaches do not maintain explicit
bumping information which implies that \Id{find}
has to search all levels.  In contrast,
\emph{filtered retrieval} \cite{MSSZ14} stores
several bits of bumping information per element
which is fast and simple but requires more space
than the per-block bumping information of Slick.
In some sense most similar to Slick is \emph{bumped ribbon
retrieval (BuRR)} \cite{DHSW22} which is not a
hash table, but a static retrieval data
structure whose construction relies on solving
linear equation systems.

\emph{Seperator hashing}
\cite{larson1984file,gonnet1988external,larson1988linear}
is an approach to external memory hashing that
stores per bucket bumping information similar to
the thresholds of Slick. However, by rehashing
bumped elements into the main table, overloading
cannot be used.  \hpfrage{Linear separator hashing
  \cite{larson1988linear} moves bumped elements to
  the next bucket. That way, deletions are a
  linear sweep similar to Slick.}\psfrage{discuss, I do not understand \cite{larson1988linear}.}  Also, Slick's
approach of allowing blocks to slide allows a more
flexible use of available storage.

\emph{Backyard cuckoo hashing} hashes elements to
statically allocated blocks (called bins there)
and bumps some elements to a backyard when these
blocks overflow. There is no bumping metadata. The
backyard is a cuckoo hash table whose insertion
routine is modified. It tries to terminate cuckoo
insertion by moving elements back to the primary
table. When plugging in concrete values, the space
efficiency of this approach suffers from a large
number of empty table entries.  For example, to
achieve space overhead below 10 \%, this approach
uses blocks of size at least $333$.
%% An
%% interesting possibility for future work is whether
%% we can improve that by combining a Slick primary
%% table with the backyard management from backyard
%% cuckoo hashing therby avoiding our backyard
%% cleaning rounds.

\emph{Iceberg hashing} also
hashes elements to statically allocated
blocks. Metadata counts overflowing elements but
this implies that \Id{find} still has to search both
main table and backyard for overflowing blocks.
Iceberg hashing also needs much larger blocks than
Slick since no sliding or other balancing
approaches besides bumping are used.  For example,
a practical variant of iceberg hashing
\cite{pandey2022iceberght} uses $B=64$ and still
has 15 \% empty cells.  A theoretical variant that
achieves succinctness \cite{bender2021all} uses
blocks of size $\Oh{\log^2n}$ and uses complex
metadata inside blocks to allow constant time
search.\psfrage{HP mentions
  \cite{awad2023analyzing}. Should we also
  cite/discuss that one?}

There has been intensive further
theoretical work on achieving succinctness and
constant time operations.  We refer to
\cite{bender2021all,BenderFKKL22} for a recent
overviews. Slick does not achieve all the features
mentioned in these results, e.g., with respect to
stability or high probability bounds. However,
Slick is not only simpler and likely more
practical than these approaches, but may also
improve some of the crucial bounds.  For example,
the main result in
\cite{BenderFKKL22} is a tradeoff between
query time proportional to a parameter $k$ and per
element space overhead of $\Ohsmall{\log^{(k)}n}$ bits where
Slick achieves a similar effect by simply choosing
an appropriate block size without direct impact on
query time. Future (theoretical) work on Slick
could perhaps achieve $\Oh{1}$ bits of overhead per
element by exploring the remaining flexibility in
arranging elements within a block that can encode
$\log B-\Oh{1}$ bits of information per element
using a standard trick of implicit algorithms.
\hpfrage{This feels like it doesn't really belong into the related work section. Move to succinct section?}

\emph{PaCHash} \cite{kurpicz2023pachash} is a
static hash table that allows storage of the
elements without any gaps using a constant number
of metadata bits per block. This even works for
objects of variable size.  It seems difficult
though to make PaCHash dynamic and
PaCHash-\Id{find} is likely to be slower as it
needs predecessor search in an Elias-Fano encoded
monotone sequence of integers in addition to
scanning a block of memory.

\psfrage{also discuss more stuff cited in PaCHash?}

\psfrage{follow references in BuRR?}

\psfrage{folly?}

A technique resembling the sliding approach in
Slick are sparse tables used for representing
sorted sequences \cite{IKR81}.  Here a sorted
array is made dynamic by allowing for empty cells in
the array.  Insertion is rearranging the elements
to keep the gaps well distributed.  Slick can
slide faster as it is not bound to keep them
sorted and bumping further reduces the high
reorganization overhead of sparse tables.

%%%%%%%%%%%%%%%%%%%%%%%%%%%%%%%%%%%%%%%%%%%%%%%%%%%%%%%%%%%%%%%%%%%%%%
\section{Conclusions and Future Work}\label{s:conclusion}

With Slick (and its variants Luckoo and blocked/bumped Robin
Hood), we have described an approach to obtain
hash tables that are simple, fast and space
efficient. We are in the process of implementing
and analyzing Slick.  This report already outlines
a partial analysis but we need to get a more
concrete grip on how the insertion time depends on
the number of empty slots.

Several further algorithmic improvements and
applications suggest themselves.  In particular,
we believe that Slick can be adapted to be space
efficient also when the final size of the table is
not known in advance. We believe that Slick can be
used to implement a space efficient approximate
membership query filter (AMQ, aka Bloom filter).
Concretely, space efficient AMQs can be
represented as a set of hash values (this can also
be viewed as a \emph{compressed single shot Bloom
filter}) \cite{PSS10}. The successful dynamic
\emph{quotient filter} AMQ
\cite{PPR05,bender2012don,maier2020concurrent} can
be viewed as an implementation of this idea
using Cleary's compact hashing \cite{Cle84}.  Doing this
with Slick instead promises a better
space-performance tradeoff.

On the practical side, an interesting question is
whether Slick could help to achieve better
hardware hash tables, as its simple \Id{find}
function could in large parts be mapped to
hardware (perhaps with a software handled FAIL for
bumped elements).  Existing hardware hash tables
(e.g., \cite{FAK16}) seem to use a more rigid
non-slidable block structure.

\myparagraph{Acknowledgements.}  The authors would
like to thank Peter Dillinger for early
discussions eventually leading to this paper.
This project has received funding from the
European Research Council (ERC) under the European
Union’s Horizon 2020 research and innovation
programme (grant agreement No. 882500). Stefan
Walzer is funded by the Deutsche
Forschungsgemeinschaft (DFG, German Research
Foundation) 465963632.

\begin{center}
  \includegraphics[width=4cm]{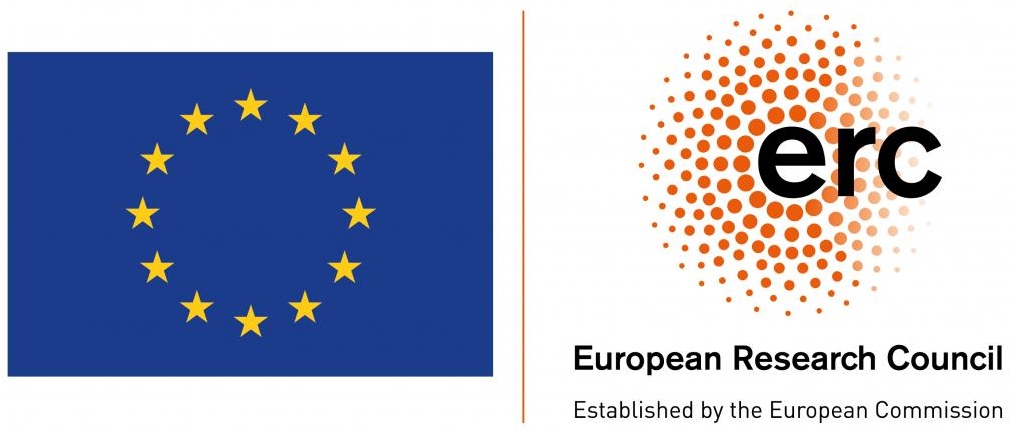}
\end{center}

% Possible directions for future research:
% \psfrage{TBD: We can use table lookups to decode $\Oh{\log_a m}$ subsequent offsets in constant time. Use that in the analysis? This would be even more attractive if we can use a generic code for normally distributed numbers.}
% \hpfrage{Why is the bit vector so periodic?}
% \hpfrage{Dillinger: Use rice code to encode $p$?}
% \hpfrage{Quotienting: Compress keys using a random permutation together with the bin number}

\bibliographystyle{plainurl}
\bibliography{diss}

\end{document}